\documentclass[conference]{IEEEtran}
\IEEEoverridecommandlockouts
\usepackage{cite}
\usepackage{amsmath,amssymb,amsfonts}
\usepackage{algorithmic}
\usepackage{graphicx}
\usepackage{textcomp}
\usepackage{subfigure}
\usepackage{array}
\usepackage{multirow} 
\allowdisplaybreaks[4]
\usepackage{xcolor}
\def\BibTeX{{\rm B\kern-.05em{\sc i\kern-.025em b}\kern-.08em
    T\kern-.1667em\lower.7ex\hbox{E}\kern-.125emX}}
\usepackage{ntheorem}
\theorembodyfont{\upshape}
\theoremheaderfont{\it}
\qedsymbol{\square}
\theoremseparator{:}
\newtheorem{theorem}{\indent Theorem}
\newtheorem{lemma}{\indent Lemma}
\newtheorem*{proof}{\indent Proof}
\newtheorem{remark}{\indent Remark}
\newtheorem{corollary}{\indent Corollary}

\begin{document}

\title{Digital versus Analog Transmissions for Federated Learning over Wireless Networks
 \thanks{This work was supported in part by National Key R\&D Program of China (Grant No. 2023YFB2904804), the NSFC under grants 62022026~and 62211530108, and the Fundamental Research Funds for the Central~Universities 2242022k60002 and 2242023K5003. (Corresponding author: Wei~Xu)}
}
\vspace{-0.5cm}
\author{\IEEEauthorblockN{Jiacheng~Yao\IEEEauthorrefmark{1},
Wei~Xu\IEEEauthorrefmark{1}\IEEEauthorrefmark{2},
Zhaohui~Yang\IEEEauthorrefmark{3},
Xiaohu~You\IEEEauthorrefmark{1}\IEEEauthorrefmark{2}, Mehdi~Bennis\IEEEauthorrefmark{4}, and
H. Vincent Poor\IEEEauthorrefmark{5}}
\IEEEauthorblockA{\IEEEauthorrefmark{1}National Mobile Communications Research Laboratory, Southeast University, Nanjing, China}
\IEEEauthorblockA{\IEEEauthorrefmark{2}Purple Mountain Laboratories, Nanjing, China}
\IEEEauthorblockA{\IEEEauthorrefmark{3}College of Information Science and Electronic Engineering, Zhejiang University, Hangzhou, China}
\IEEEauthorblockA{\IEEEauthorrefmark{4}Center for Wireless Communications, University of Oulu, Oulu, Finland}
\IEEEauthorblockA{\IEEEauthorrefmark{5}Department of Elect. \& Computer Eng., Princeton University, Princeton, NJ, USA}
\IEEEauthorblockA{Emails: \{jcyao,wxu\}@seu.edu.cn, yang\_zhaohui@zju.edu.cn}
\vspace{-0.9cm}}

\maketitle

\begin{abstract}
In this paper, we quantitatively compare these two effective communication schemes, i.e., digital and analog ones, for wireless federated learning (FL) over resource-constrained networks, highlighting their essential differences as well as their respective application scenarios. We first examine both digital and analog transmission methods, together with a unified and fair comparison scheme under practical constraints.  A universal convergence analysis under various imperfections is established for FL performance evaluation in wireless networks. These analytical results reveal that the fundamental difference between the two paradigms lies in whether communication and computation are jointly designed or not. The digital schemes decouple the communication design from specific FL tasks, making it difficult to support simultaneous uplink transmission of massive devices with limited bandwidth. In contrast, the analog communication allows over-the-air computation (AirComp), thus achieving efficient spectrum utilization. However, computation-oriented analog transmission reduces power efficiency, and its performance is sensitive to computational errors.  Finally, numerical simulations are conducted to verify these theoretical observations.
\end{abstract}


\vspace{-0.2cm}

\section{Introduction}
To support emerging intelligent applications in anticipated sixth-generation (6G) mobile networks, federated learning (FL) is viewed as a promising distributed learning technique that has garnered significant attention from academic and industrial circles, primarily due to its communication-efficient and privacy-enhancing characteristics \cite{6g1,mzchen,xu,xujindan}.  Recent studies have explored the implementation of FL algorithms at the wireless edge \cite{wfl1, energy, wf2}. However, limited communication resources pose a significant bottleneck to the performance of wireless FL \cite{mz2}. One particular concern regards the uplink transmission process, where numerous participating devices need to transmit local updates to the parameter sever (PS), leading to a substantial increase in communication overhead and latency \cite{mz3}. Hence, the development of an efficient uplink transmission scheme is crucial to enable wireless FL.

Digital communication schemes have been widely considered in recent works to enable the wireless FL \cite{quan1,quan2,gomore}, where the local updates are quantized into finite bits and then transmitted to the PS via traditional multiple access methods. At the receiver, the PS relies on channel coding for error detection and correction, before model aggregation based on the received accurate local updates. In addition to the traditional digital communications, analog communication paradigm is an alternative communication-efficient way for the implementation of wireless FL \cite{poor,zhu}. In particular, the local updates are amplitude-modulated and then transmitted reusing the available radio resource simultaneously. Due to the superposition property of radio channels, the global model can be computed automatically over-the-air, which is therefore referred to as over-the-air computation (AirComp). Analog communication pushes model aggregation from the PS to the air, which not only functionally but physically integrates the computation and communication and substantially reduces the communication latency of wirless FL.

In general, by incorporating learning task-oriented design, both digital and analog transmissions are effective ways to fulfill the communication requirements of wireless FL.  In traditional communication for data transmission, digital communication schemes have been proven not only in theory but also in practice as dominatively outperforming analog techniques. In communications for computation tasks, however, analog communication has shown to be exceptionally effective in some cases of resource-constrained networks \cite{comp}. Hence, it is of interest to comprehensively compare the digital and analog transmissions for wireless FL. However, to the best of our knowledge, there is a lack of literature that presents a comprehensive comparison between the two fundamental communication paradigms, especially over the resource-constrained wireless FL networks. Also, there have been few attempts to elucidate the fundamental differences, which is crucial for effective deployment and detailed design.

Against this background, in this paper, we propose a unified framework for digital and analog transmissions in wireless FL and conduct a theoretical comparison between them with strict latency requirements and transmit power budgets. To accurately evaluate the performance, we establish a universal convergence performance analysis framework, which captures the impact of various factors, including additional noise, partial participation, coefficient distortion, and non-independent and non-identically distributed (non-IID) datasets. Analytical results reveal that the digital schemes are difficult to achieve satisfactory performance with the increasing edge devices under limited radio resources. In contrast, the analog scheme exhibits a higher level of efficiency in spectrum utilization. However, the introduction of computation goals in the communication process results in less efficient transmit power utilization, and the presence of channel state information (CSI) uncertainties inevitably comes with computation errors.

\section{Federated Learning Model}

We consider a typical wireless FL system, where $K$ distributed devices are coordinated by a central PS to perform FL. 
Let $\mathcal{D}_k$ denote the local dataset owned by the $k$-th device, which contains $D_k=\vert \mathcal{D}_k \vert$ training samples. The goal of the FL algorithm is to find the optimal $d$-dimensional model parameter vector, denoted by $\mathbf{w}^*\in \mathbb{R}^{d\times 1}$, to minimize the global loss function $F(\mathbf{w})$, i.e., 

\vspace{-0.5cm}
\begin{align}\label{eq1}
\mathbf{w}^* =\arg \min_{\mathbf{w}} F(\mathbf{w})=\arg \min_{\mathbf{w}} \sum_{k=1}^K \alpha_k F_k(\mathbf{w}),
\end{align}
where $\alpha_k$ represents the aggregation weight for device $k$, and $F_k(\mathbf{w})$ is the local loss function at device $k$ defined by
\begin{align}
F_k (\mathbf{w})= \frac{1}{D_k} \sum_{\mathbf{u}\in \mathcal{D}_k}	\mathcal{L}(\mathbf{w},\mathbf{u}),
\end{align}
where $\mathbf{u}$ denotes a training sample selected from $\mathcal{D}_k$, and $\mathcal{L}(\mathbf{w},\mathbf{u})$ represents the sample-wise loss function with respect to $\mathbf{u}$. We note that the different local datasets at distinct devices are usually non-IID, and the optimal model parameters in (\ref{eq1}) are not necessarily the optimal for local datasets. Let $\mathbf{w}_k^*$ denote the locally optimal model parameter vector at device $k$ and it is usually different from  $\mathbf{w}^*$. 

To handle (\ref{eq1}), an FL algorithm performs the model training iteratively. In the $m$-th round of the FL algorithm, firstly, the PS broadcasts the latest global model $\mathbf{w}_m$ to all the devices.  After receiving $\mathbf{w}_m$, each device exploits its local dataset to compute the local gradient as
\begin{align}\label{eq3}
\mathbf{g}_m^k\triangleq \nabla F_k(\mathbf{w}_{m})= \frac{1}{D_k} \sum_{\mathbf{u}\in \mathcal{D}_k}	\nabla \mathcal{L}(\mathbf{w}_m,\mathbf{u}), \enspace \forall k.
\end{align}
Upon receiving all the local gradients reported by each device, the PS calculates the estimated global gradient as 

\vspace{-0.5cm}
\begin{align}\label{gradient}
\mathbf{g}_m \triangleq \sum_{k=1}^K \alpha_k \mathbf{g}_m^k,
\end{align}
and then updates the global model as $\mathbf{w}_{m+1}=\mathbf{w}_m-\eta \mathbf{g}_m$, where $\eta$ is the learning rate. The above steps iterate until a convergence condition is met.

Considering a sufficient power budget at the PS, the downlink transmission is usually assumed error-free \cite{wfl1}. Otherwise, for uplink transmission with limited communication resources, additional errors are inevitable. Moreover, considering the limited resources in practice, only partial devices are activated at each round. Let $\mathcal{S}_m$ denote the set of the activated devices in round $m$ and $N=\vert \mathcal{S}_m\vert$ be the number of participating devices. We assume that the PS performs non-uniform sampling without replacement to select the participating devices. Then, we denote the inclusion probability of device $k$ by $r_k$, which represents probability of device $k$ being selected per round.

\section{Uplink Transmission and Comparison Method}

In wireless FL, we rely on the imperfect uplink transmission to provide an estimation of the gradient in (\ref{gradient}). Assuming that total uplink bandwidth $B$ is divided into $M$ subbands, which supports orthogonal access for up to $M(M\geq N)$ devices. 
Let $\bar{h}_k =L_k h_k$  be the channel between the $k$-th device and the PS, where $L_k$ denotes the large-scale path loss, and $h_k\sim \mathcal{CN}(0,1)$ represents the small-scale fading of the channel.
In practice, due to fast varying characteristic, perfect estimation of the small-scale fading channel is usually not available at the devices. Let $\hat{h}_k$ denote the channel estimate at device $k$. Then, we model the CSI imperfection of the small-scale fading as $h_k = \rho \hat{h}_k + \sqrt{1-\rho^2}v_k$, where $\rho\in(0,1]$ is correlation coefficient between $h_k$ and $\hat{h}_k$ and reflects the level of channel estimation accuracy, and $v_k\sim \mathcal{CN}(0,1)$ is the CSI error independent of $\hat{h}_k$. 

\subsection{Digital Transmission Method}
In the digital transmission, the $N$ selected devices first quantize their gradients into a finite number of $b$ bits and then transmit the quantized local updates to the PS. Specifically, we adopt the stochastic quantization method to quantize the local updates \cite{quan2}. 
Denote the maximum and the minimum values of the modulus among all parameters in $\mathbf{g}_m^k$ by $g_{\max}$ and $g_{\min}$, respectively. Then, the interval $[g_{\min}, g_{\max}]$ is divided evenly into $2^b-1$ quantization intervals denoted by $\tau_i$ for $i\in[2^b-1]$. Given $\vert x \vert \in [\tau_i,\tau_{i+1})$, the quantization function $\mathcal{Q}(x)$ is
\begin{align}\label{quan}
\mathcal{Q}(x)=\left\{ \begin{array}{cc}
\mathrm{sign}(x)\tau_i &\mathrm{w.p.} \enspace \frac{\tau_{i+1}-\vert x\vert}{\tau_{i+1}-\tau_i},\\
\mathrm{sign}(x)\tau_{i+1} &\mathrm{w.p.} \enspace \frac{\vert x\vert-\tau_i}{\tau_{i+1}-\tau_i},\\
\end{array} \right.
\end{align}
where $\mathrm{sign}(\cdot)$ represents the signum function and ``w.p." represents ``with probability." Exploiting the quantization function in (\ref{quan}), the local update $\mathbf{g}_m^k$ is quantized as $\mathcal{Q}\left(\mathbf{g}_m^k \right)\triangleq \left[\mathcal{Q}\left(g_{m,1}^k \right),\cdots,\mathcal{Q}\left(g_{m,d}^k \right)\right]^T$. Let $q$ denote the number of bits used to represent $g_{\min}$ and $g_{\max}$. Hence, the total number of bits needed to transmit the quantized gradients amounts to $b_{\mathrm{total}}=d(b+1)+q$.

Then, during the uplink transmission, each device occupies different carriers equally to avoid interference with each other. Hence, the channel capacity of device $k$ equals to
\begin{align}
C_k=B_k \log_2\left( 1+\frac{P_k\vert \bar{h}_k\vert^2}{B_k N_0} \right),
\end{align}
where $B_k=B/N$ is the bandwidth allocated to device $k$, $P_k$ is the transmit power at device $k$, and $N_0$ is the noise power density. To avoid the uncontrolled severe delay brought by stragglers, we assume that all devices transmit the local updates at a fixed rate $R=\frac{B}{N}\log_2(1+\theta)$, where $\theta$ is a chosen constant \cite{gomore}. According to \cite{wfl1}, the transmission is assumed error-free if $R\leq C_k$. Hence, the probability of successful transmission at device $k$ is calculated as
\begin{align}\label{eq9}
p_k=\Pr\left\{ R\leq C_k \right \}=\mathrm{exp} \left (-\frac{BN_0 \theta}{2NP_k L_k^{2}} \right).
\end{align}
At the PS, a cyclic redundancy check (CRC) mechanism is applied to check the detected data such that erroneous gradients can be excluded from the model aggregation. Finally, the obtained estimate of the desired gradient in (\ref{gradient}) is given by
\begin{align}\label{estd}
\hat{\mathbf{g}}_{m,\text{D}}= \sum_{k=1}^K \frac{\chi_k \alpha_k \xi_{k,\text{D}}}{r_k}\mathcal{Q}(\mathbf{g}_m^k),
\end{align}
where $\chi_k$ is an indicator variable for the device selection, $\frac{1}{r_k}$ is to guarantee the unbiasedness of gradient estimate, and $\xi_{k,\text{D}}$ represents distortion brought by packet loss. To be concrete, $\chi_k$ is 1 if $k\in \mathcal{S}_m$ and otherwise $\chi_k$ is 0. The distortion $\xi_{k,\text{D}}$ is characterized by the probability in (\ref{eq9}) as
\begin{align}\label{distortiond}
\xi_{k,\text{D}}=\left \{  
\begin{array}{cl}
\frac{1}{p_k} &\text{w.p.}\enspace p_k,\\
0&\text{w.p.}\enspace 1-p_k,
\end{array}
\right.
\end{align}
where $\xi_{k,\text{D}}$ satisfies that $\mathbb{E}\left[ \xi_{k,\text{D}}\right]=1$. With the gradient estimate in (\ref{estd}), the updated global model at the $(m+1)$-th round is reformulated as $\tilde{\mathbf{w}}_{m+1}= \tilde{\mathbf{w}}_{m}-\eta\hat{\mathbf{g}}_{m,\text{D}}$,
where $\tilde{\mathbf{w}}_{m}$ denotes the model obtained at the previous round.

\subsection{Analog Transmission Method}
In the analog transmission with AirComp, the selected devices simultaneously upload the uncoded analog signals of the local gradients to the PS and a weighted summation of the local updates in (\ref{gradient}) can be achieved over-the-air. Specifically, the received signal at the PS is expressed as
\begin{align}\label{eq13}
\mathbf{y}=\sum_{k=1}^K \chi_k\bar{h}_k  \beta_k \mathbf{g}_m^k +\mathbf{z}_m,
\end{align} 
where $\beta_k$ is the pre-processing factor at device $k$, and $\mathbf{z}_m$ is additive white Gaussian noise following $\mathcal{CN}(\mathbf{0},BN_0\mathbf{I})$. To accurately estimate the desired gradient in (\ref{gradient}), the pre-processing factor $\beta_k$ should be adapted to the channel coefficients based on the CSI available at each device. We adopt the typical truncated channel inversion scheme to combat the deep fades\cite{zhu}. It is expressed as
\begin{align}
    \beta_k=\left \{  
\begin{array}{cl}
\frac{\zeta \lambda \alpha_k \hat{h}_k^*}{ r_k L_k \vert \hat{h}_k \vert^2} & \vert \hat{h}_k \vert^2 \geq \gamma_{\mathrm{th}},\\
0&\vert \hat{h}_k \vert^2 <\gamma_{\mathrm{th}},
\end{array}
\right.
\end{align}
where $\gamma_{\mathrm{th}}$ is a predetermined power-cutoff threshold, $\zeta$ is a scaling factor for ensuring the transmit power constraint, and compensation coefficient $\lambda$ is selected to alleviate the impact of imperfect CSI \cite{imperfect}. 

At the receiver, the PS scales the real part of $\mathbf{y}$ in (\ref{eq13}) with $\frac{1}{\zeta}$ and obtain an estimate of the actual gradient in (\ref{gradient}). It yields
\begin{align} \label{ghata}
\hat{\mathbf{g}}_{m,\text{A}}=\sum_{k=1}^K \frac{\chi_k \alpha_k \xi_{k,\text{A}} }{r_k } \mathbf{g}_m^k+\bar{\mathbf{z}}_m,
\end{align}
where $\bar{\mathbf{z}}_m\triangleq \frac{\Re\{\mathbf{z}_m\}}{\zeta}$ is the equivalent noise, and $\xi_{k,\text{A}}$ denotes the distortion brought by the analog transmission with imperfect CSI. It follows
\begin{align}
\xi_{k,\text{A}} =\left \{  
\begin{array}{cl}
\frac{\lambda\Re\{h_k^* \hat{h}_k\}}{ \vert \hat{h}_k \vert^2} &  \text{w.p.} \enspace \mathrm{e}^{-\gamma_{\mathrm{th}}},\\
0&\text{w.p.}\enspace 1-\mathrm{e}^{-\gamma_{\mathrm{th}}}.
\end{array}
\right.
\end{align}
Similarly, the global model at the $(m+1)$-th round under the analog transmission is updated as
$\tilde{\mathbf{w}}_{m+1}= \tilde{\mathbf{w}}_{m}-\eta\hat{\mathbf{g}}_{m,\text{A}}$.

\subsection{A Unified Framework for Wireless FL Comparison}
To minimize the performance loss brought by imperfect uplink transmission, the overall FL task-oriented optimization over the wireless networks can be formulated as
\begin{align} \label{p1}
    \min \quad&\mathbb{E}\left[F(\mathbf{w}_{m+1})\right]-F(\mathbf{w}^*)\nonumber \\
    \text{s.t.}\quad  & \mathrm{C}_1:T\leq T_{\max},\nonumber \\
    & \mathrm{C}_2:P_k\leq P_{\max},\enspace \forall k,
\end{align}
where $T$ represents uplink transmission delay per round, $T_{\max}$ and $P_{\max}$ denote the maximum transmission delay target and the transmit power unit, respectively. 

For a fair comparison between the two transmission paradigms, we measure the achievable optimality gap in (\ref{p1}) under the same transmission delay target and transmit power budget. Specific constraints for the two transmission paradigms are listed as follows. For digital transmission, the transmission delay per communication round is calculated as $T_{\text{D}}\!=\!\frac{b_{\mathrm{total}}}{R}\!=\! \frac{Nd(b+1)}{B\log_2(1+\theta)}$. Hence, constraint $\mathrm{C}_1$ is reformulated as $\theta\geq 2^{\frac{Nd(b+1)}{BT_{\max}}}-1$. For constraint $\mathrm{C}_2$, due to its interference-free characteristic, full power transmission is optimal and hence the constraint is reformulated by $P_k=P_{\max}, \enspace \forall k$.

For analog transmission, according to \cite[Eq. (16)]{deploy}, the per-round delay for the analog transmission follows $T_{\text{A}}=dM/B$, which is a constant. For feasibility, we assume that the target $T_{\max}$ cannot be smaller than $T_{\text{A}}$. The maximum power constraint $\mathrm{C}_2$ is rewritten as $\max_{m,k}\left \{ \left \Vert \beta_k \mathbf{g}_m^k \right \Vert^2 \right \} \leq P_{\max}$ for the analog transmission. Unlike the digital transmission, it is not possible to fully utilize the maximum power in analog transmission due to the need for channel pre-equalization.

\section{Comparison with Convergence Analysis}
In this section, we analyze the convergence performance under the digital and analog transmissions with the practical constraints for wireless FL. Based on the derived results, we further conduct detailed comparisons between the two paradigms from various perspectives of view.
\subsection{Preliminaries}

To pave the way for performance analysis, we provide necessary assumptions and lemmas about the learning algorithms and the transmission paradigms. Firstly, we make several common assumptions on the loss functions, which are widely used in FL studies, like \cite{wfl1,importance}.

\emph{Assumption 1}: The local loss functions $F_k(\cdot)$ are $\mu$-strongly convex for any device $k$.

\emph{Assumption 2}: The local loss functions $F_k(\cdot)$ are differentiable and have $L$-Lipschitz gradients.

\emph{Assumption 3}: In most practical applications, it is safe to assume that the sample-wise gradient is always upper bounded by a finite constant $\gamma$, i.e., $\left \Vert \nabla \mathcal{L}(\mathbf{w},\mathbf{u}) \right \Vert \leq \gamma$.

\emph{Assumption 4}: The distance between the locally optimal model, $\mathbf{w}_k^*$, and the globally optimal model, $\mathbf{w}^*$, is uniformly bounded by  a finite constant $\delta$, i.e., $\Vert \mathbf{w}_k^*- \mathbf{w}^* \Vert\leq \delta$, which directly captures the degree of non-IID.

Building upon the above assumptions, we immediately conclude that the global loss function $F(\cdot)$ is also $\mu$-strongly convex and $L$-smooth. Then, we provide the following lemma regarding the imperfection in digital and analog transmissions.

\begin{lemma}
    Under the stochastic quantization and the proposed digital aggregation in (\ref{estd}), $\hat{\mathbf{g}}_{m,\text{D}}$ is an unbiased estimate of the actual gradient in (\ref{gradient}). For the considered analog paradigm in (\ref{ghata}), by choosing $\lambda = \frac{e^{\gamma_{\mathrm{th}}}}{\rho}$, the gradient estimate $\hat{\mathbf{g}}_{m,\text{A}}$ is also unbiased.
\end{lemma}

\begin{proof}
    For the digital case, according to \emph{Lemma 5} in \cite{quan1}, we first conclude that the quantized gradients $\mathcal{Q}(\mathbf{g}_m^k)$ is unbiased, i.e., $\mathbb{E}\left [ \mathcal{Q}(\mathbf{g}_m^k)\right ]=\mathbf{g}_m^k$. Combining with the fact that $\mathbb{E}\left[\xi_{k,\text{D}}\right]=1$ in (\ref{distortiond}), we have
    \begin{align}\label{eq47}
        \mathbb{E}\left[\hat{\mathbf{g}}_{m,\text{D}}\right]=\sum_{k=1}^K \alpha_k
        \mathbb{E}\left[\frac{\chi_k}{r_k}\right]\mathbb{E}\left[\xi_{k,\text{D}}\right]\mathbb{E}\left [ \mathcal{Q}(\mathbf{g}_m^k)\right ]=\mathbf{g}_m.
    \end{align}
    As for the analog transmission, by directly exploiting the results in \emph{Lemma 1} of our previous work \cite{imperfect}, we have $\mathbb{E}\left[ \xi_{k,\text{A}}\right]=1$. Combining with the statistical characteristic of $\chi_k$ and $\bar{\mathbf{z}}_m$ and following the same procedures in (\ref{eq47}), we have $\mathbb{E}\left[\hat{\mathbf{g}}_{m,\text{A}}\right]=\mathbf{g}_m$. The proof completes. \hfill $\square$
\end{proof}

Although both the digital and analog transmissions achieve unbiased gradient estimations, there are fundamental differences in the distortion between the two paradigms. For the digital transmission, the distortion mainly lies in the gradients themselves, i.e., gradient quantization errors. On the other hand, due to the integration of communication and computation in AirComp, the analog transmission additionally suffers from computation errors due to the imperfect channel alignment. This essential difference further discriminates the performances of the two paradigms in specific scenarios, which are elaborated in the following.

\subsection{Convergence Analysis and Detailed Comparison}

Based on the \emph{Assumptions 1-4} and \emph{Lemma 1}, we characterize the convergence performance under different transmission paradigms in the following theorem.

\begin{figure*}[t]
\begin{align}
\mathbb{E}&\left[ F(\tilde{\mathbf{w}}_{m+1})\right]-F(\mathbf{w}^*) \leq \frac{L}{2}\left( 1-\eta \mu +2\eta^2 L^2g_{\text{D}}(\mathbf{r},b)\right)^{m+1} \mathbb{E}\left[\left \Vert \tilde{\mathbf{w}}_{0} -\mathbf{w}^*\right \Vert^2 \right] +\frac{\eta (L\phi(b) +2L^3\delta^2) g_{\text{D}}(\mathbf{r},b)}{2\mu-4\eta L^2 g_{\text{D}}(\mathbf{r},b)},\label{digital}\\
\mathbb{E}&\left[ F(\tilde{\mathbf{w}}_{m+1})\right]-F(\mathbf{w}^*)\leq \frac{L}{2}\left( 1-\eta \mu +2\eta^2 L^2g_{\text{A}}(\mathbf{r},\gamma_{\mathrm{th}})\right)^{m+1} \mathbb{E}\left[\left \Vert \tilde{\mathbf{w}}_{0} -\mathbf{w}^*\right \Vert^2 \right] +\frac{\eta \left(L\varphi(\mathbf{r},\gamma_{\mathrm{th}})+2L^3 \delta^2 g_{\text{A}}(\mathbf{r},\gamma_{\mathrm{th}}) \right)}{2\mu-4\eta L^2 g_{\text{A}}(\mathbf{r},\gamma_{\mathrm{th}})}. \label{analog}
\end{align}
\hrulefill
\vspace{-0.4cm}
\end{figure*}

\begin{theorem}
For a fixed learning rate satisfying $\eta \leq \frac{\mu}{2L^2 g_{\text{D}}(\mathbf{r},b)}$, the optimality gap of the distributed gradient update in the $(m+1)$-th iteration under the digital transmission is equal to (\ref{digital}) at the top of the next page, where $\phi(b)$ is a constant regarding the quantization errors, $\mathbf{r}\triangleq [r_1,\cdots,r_K]^T$, and $g_{\text{D}}(\mathbf{r},b)\triangleq \sum_{k=1}^K \frac{\alpha_k}{p_k r_k}$. 

For a fixed learning rate satisfying $\eta \leq \frac{\mu}{2L^2 g_{\text{A}}(\mathbf{r},\gamma_{\mathrm{th}})}$, the optimality gap of the distributed gradient update in the $(m+1)$-th iteration under the analog transmission is equal to (\ref{analog}) at the top of the next page, where $g_{\text{A}}(\mathbf{r},\gamma_{\mathrm{th}})\triangleq \sum_{k=1}^K \frac{\alpha_k}{r_k}\left(e^{\gamma_{\mathrm{th}}}+\frac{(1-\rho^2)\mathrm{E}_1\left(\gamma_{\mathrm{th}} \right)e^{2\gamma_{\mathrm{th}}}}{2\rho^2}\right)-1$ and $\varphi(\mathbf{r},\gamma_{\mathrm{th}})\triangleq \frac{B N_0 \gamma^2 e^{2\gamma_{\mathrm{th}}}}{2P_{\max}\rho^2 \gamma_{\mathrm{th}}}\max_k\left \{\frac{\alpha_k^2}{r_k^2 L_k^2}\right\}$.
\end{theorem}

\begin{proof}
Please refer to the Appendix. \hfill $\square$
\end{proof}

According to \emph{Theorem 1},  we are ready to derive the optimality gap after convergence in the following corollary, which reflects the ultimately achievable convergence performance.

\begin{corollary}
With sufficient rounds, the optimality gap achieved by digital and analog transmissions respectively converges to
\begin{align}
G_{\text{D}}&=\frac{\eta (L \phi(b) +2L^3\delta^2) g_{\text{D}}(\mathbf{r},b)}{2\mu-4\eta L^2 g_{\text{D}}(\mathbf{r},b)},\nonumber \\
G_{\text{A}}&= \frac{\eta \left(L\varphi(\mathbf{r},\gamma_{\mathrm{th}})+2L^3 \delta^2 g_{\text{A}}(\mathbf{r},\gamma_{\mathrm{th}}) \right)}{2\mu-4\eta L^2 g_{\text{A}}(\mathbf{r},\gamma_{\mathrm{th}})}.
\end{align}
\end{corollary}

From \emph{Corollary 1}, we further compare the two paradigms from the following perspectives and conclude the insightful remarks that are instructive for the deployment of FL in practice.
Without loss of generality, we drop the unbalance of the datasets and assume uniform inclusion probabilities, i.e., $\alpha_k =\frac{1}{K}$, and $r_k=\frac{N}{K}$, $\forall k$. Also, we set that $T_{\max}=T_{\text{A}}$. As a summary, we present the main comparison results in Table~\ref{table:2}.

\begin{table*}[!t]   
\renewcommand\arraystretch{1.1}
\begin{center}   
\caption{Main Comparison Results with Respect To Achievable Optimality Gap}  
\label{table:2} 
\begin{tabular}{|m{1.5cm}<{\centering}|m{3cm}<{\centering}|m{3cm}<{\centering}|m{4.5cm}<{\centering}|m{3cm}<{\centering}|}   
\hline   
\multirow{2}*{Paradigms} &  \multicolumn{2}{c|}{Transmit power budget, $P$} & \multirow{2}*{Device number, $N$} & \multirow{2}*{Imperfect CSI, $\rho$} \\
\cline{2-3}
&Low SNR & High SNR & &\\
\hline 
Digital & $\mathcal{O}\left( \mathrm{exp}\left(\frac{\varepsilon}{P}\right)\right)$ $\searrow$ & $\rightarrow$ $G_{\text{D}}^{\infty}$ & $\mathcal{O}\left(\frac{1}{N}\mathrm{exp}(\varepsilon_1 2^{\varepsilon_2N}/N)\right)$  $\nearrow$& / \\
\hline
Analog & $\mathcal{O}\left(\frac{1}{P}\right)$ $\searrow$ &$\rightarrow$ $G_{\text{A}}^{\infty}$ & $\mathcal{O}\left(\frac{1}{N}\right)$ $\searrow$ & $\mathcal{O}\left(1/\rho^2\right)$ $\nearrow$  \\
\hline
\end{tabular}   
\end{center} 
\vspace{-0.5cm}
\end{table*}

\subsubsection{Impact of Transmit Power}
With low signal-to-noise ratio (SNR) levels, we note that the achievable optimality gap under the digital transmission, $G_{\text{D}}$, vanishes as $\mathcal{O}\left(\mathrm{exp}\left({\varepsilon}/{P_{\max}}\right)\right)$ with the increase of the maximum transmit power budget $P_{\max}$, where $\varepsilon\triangleq \max_k \left \{ \frac{B N_0 \theta}{2N L_k^{2}}\right \}$. Especially for high SNR regime, i.e., $P_{\max}\to \infty$, $G_{\text{D}}$ tends to
\begin{align}\label{infd}
G_{\text{D}}^{\infty}\triangleq \lim_{P_{\max}\to \infty} G_{\text{D}}=\frac{\eta (L \phi(b) +2L^3\delta^2)K}{2\mu N-4\eta L^2K}.
\end{align}
On the other hand, the decay rate for $G_{\text{A}}$ is $\mathcal{O}\left({1}/{P_{\max}}\right)$ with low SNR values and the limiting value follows
\begin{align}\label{infa}
G_{\text{A}}^{\infty}\triangleq \lim_{P_{\max}\to \infty} G_{\text{A}}=\frac{2\eta L^3\delta^2\left(Kc-N\right)}{2\mu N-4\eta L^2\left(Kc-N\right)},
\end{align}
where $c\triangleq e^{\gamma_{\mathrm{th}}}+\frac{(1-\rho^2)\mathrm{E}_1\left(\gamma_{\mathrm{th}} \right)e^{2\gamma_{\mathrm{th}}}}{2\rho^2}$.

\begin{remark}
At high SNR levels, the optimality gap for the analog case comes from only the non-IID datasets while the impact of the noise asymptotically diminishes. For the digital case, however, quantization errors additionally imposes a significant impact, which can be overcome through high-precision quantization. Moreover, under the analog transmission, the negative impact of non-IID datasets is enlarged due to imperfect AirComp. Imperfect CSI results in mismatched channel inversion in AirComp, rendering perfect computation of weighted sum impossible. Hence, with inaccurate channel estimation, the ultimate performance achieved by the analog transmission may not be comparable to that by the digital case. 
\end{remark}

\subsubsection{Impact of Device Number}
With the increasing number of participating devices, $N$,  the impact of non-IID datasets asymptotically dominates the performance under the analog transmission, and $G_{\text{A}}$ decays by the order of $\mathcal{O}(1/N)$. Due to the involvement of more devices, a more accurate global gradient is obtained at the PS, which in turn facilitates the convergence of FL algorithms and leads to better performance. Meanwhile, since different devices involved in the AirComp share the same time-frequency resource, more access devices do not cause deterioration of the transmission performance, thus fully obtaining the performance gain from more devices. Hence, allowing all the active devices to participate in FL training is the optimal choice for analog transmission.

On the other hand, for the digital case, convergence performance does not necessarily monotonically change with $N$. Although more participating devices do bring performance gains, considering that limited communication resources and orthogonal access, it leads to a significant deterioration of the transmission performance, thus further affecting the convergence. Specifically, the optimality gap, $G_{\text{D}}$, enlarges with a rate of $\mathcal{O}\left(\frac{1}{N}\mathrm{exp}(\varepsilon_1 2^{\varepsilon_2 N}/N)\right)$ with sufficiently large $N$, where $\varepsilon_1 =\frac{BN_0}{2PL_K^{2}}$ and $\varepsilon_2=\frac{b+1}{M}$.
Hence, it is necessary to seek a balance between the transmission performance and diversity gain through the optimization of~$N$.

\subsubsection{Impact of Imperfect CSI}
The imperfect CSI at the transmitter only affects the performance of analog transmission, which deteriorates at the order of $\frac{1}{\rho^2}$. Due to the imperfect CSI, the aggregation computation and the truncation decision in AirComp are disturbed, thus leading to a mismatch in the model aggregation and the impact of noise amplification. 

\begin{remark}
After incorporating computation capabilities into the analog case, the emergence of computation error as a new source of error has positioned computational accuracy as a crucial factor affecting the convergence performance. It is concluded that CSI is a key factor affecting the performance gain brought by AirComp.
\end{remark}


\begin{figure*}[!t]
	\centering
	\begin{minipage}[t]{0.31\linewidth}
		\centering
		\includegraphics[width=1\linewidth]{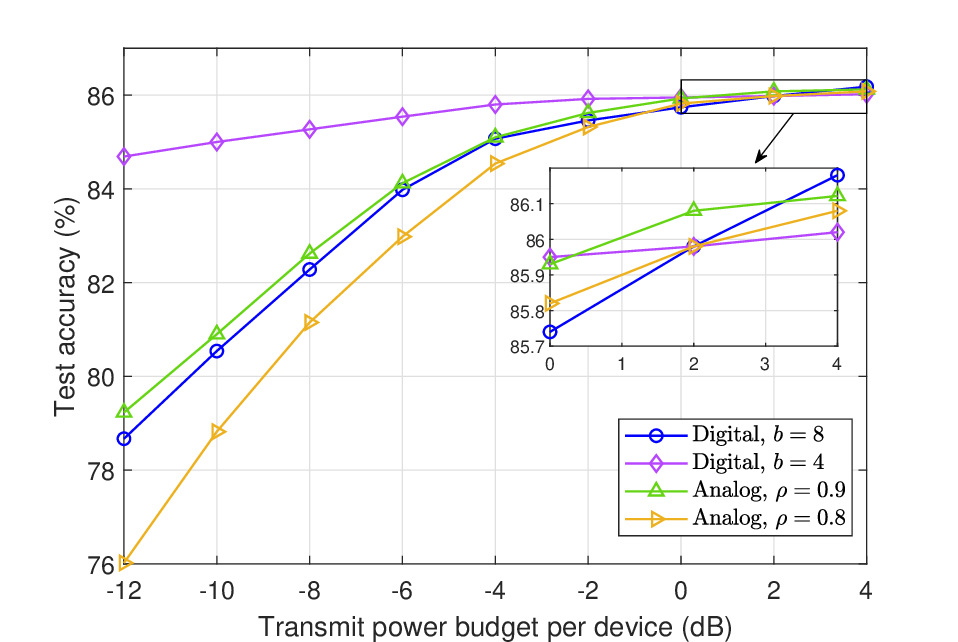}
		\caption{Test accuracy versus $P_{\max}$.}\label{fig4}
	\end{minipage}
	\begin{minipage}[t]{0.31\linewidth}
		\centering
		\includegraphics[width=1\linewidth]{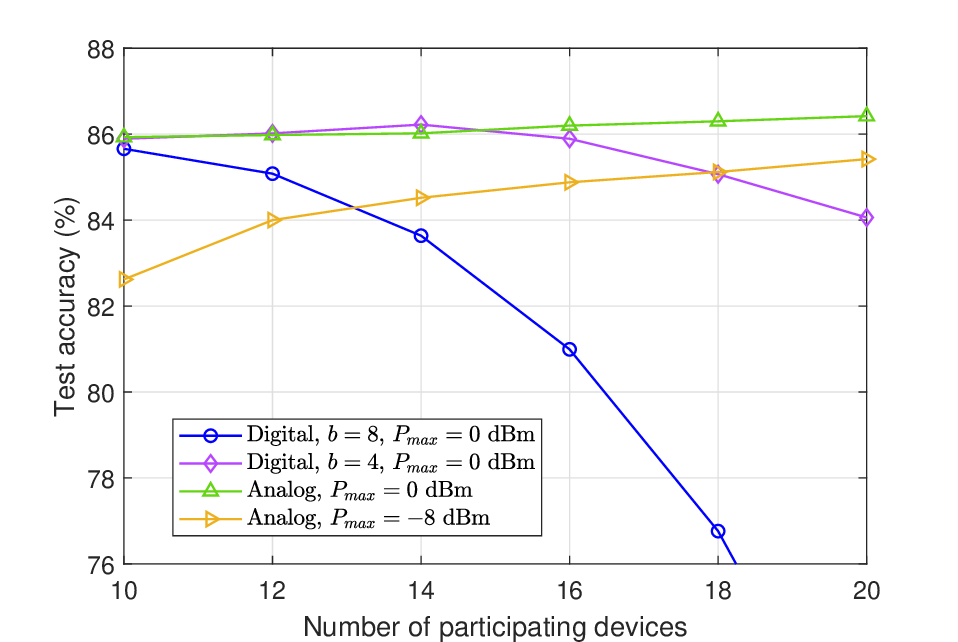}
		\caption{Test accuracy versus $N$.}\label{fig5}
	\end{minipage}
	\begin{minipage}[t]{0.31\linewidth}
		\centering
		\includegraphics[width=1\linewidth]{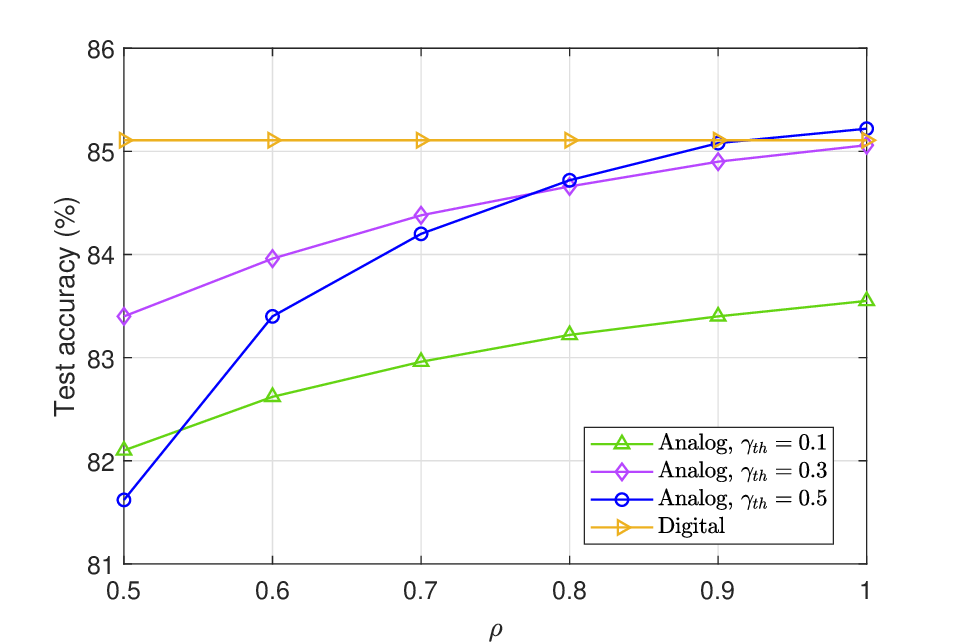}
		\caption{Test accuracy versus $\rho$.}\label{fig6}
	\end{minipage}
 \vspace{-0.4cm}
\end{figure*}

\vspace{-0.3cm}
\section{Numerical Results}
In this section, we provide simulation results to verify the performance analysis. A multi-layer perceptron (MLP) with $d=23860$ parameters is trained on MNIST dataset. Unless otherwise specified, the other parameters are set as: the number of edge devices $K=20$, the number of activated devices $N=10$, the bandwidth $B=1$ MHz,  the noise power $N_0=-80$~dBm/Hz, the maximum transmit power budget $P_{\max}=0$ dB, the number of quantization bits $b=8$, the truncation threshold $\gamma_{\mathrm{th}}=0.5$, and the learning rate $\eta=0.01$.

In Fig. \ref{fig4}, we show the test accuracy versus the maximum transmit power. It is observed that the digital transmission outperforms the analog one, particularly in situations with high SNR levels. In such cases, employing more quantization bits yields the best performance. Conversely, for low SNR levels, reducing the quantization bits leads to marginal performance loss, highlighting the flexibility of the digital schemes by selecting different quantization accuracies. On the other hand, the analog scheme faces significant performance limitations, particularly with less CSI, due to the stringent requirements of channel alignment. Hence, in terms of power utilization, the digital scheme is more efficient than the analog counterpart. 

Fig. \ref{fig5} illustrates the test accuracy versus the number of participating devices. We note that for analog transmission, the test accuracy gradually increases as $N$ increases. In contrast, although the performance in digital case may be improved initially, it eventually declines rapidly as each device can only occupy a limited amount of resources, making it unable to support high-rate transmission. Consequently, the results suggest that for digital transmission, the selection of $N$ requires further optimization, with a preference for fewer devices.

In Fig. \ref{fig6}, we present the impact of channel estimation accuracy on the analog case. It is evident that better performance can be achieved with more accurate CSI. 
Additionally, we observe that smaller truncation thresholds are more suitable for larger $\rho$, while larger truncation thresholds are preferred for smaller $\rho$. This is because higher CSI uncertainties have a significant impact on truncation choices, necessitating looser truncation conditions to reduce incorrect choices.

\section{Conclusion}

In this paper, we considered the general transmission deign of both digital and analog schemes for wireless FL, analyzed the convergence behavior of FL in terms of the optimality gap and compared the convergence from multiple perspectives. It was found that  digital scheme is more suitable for scenarios with sufficient radio resources and CSI uncertainties. On the other hand, analog scheme is suitable when their are massive numbers of participating devices.

\appendices
\section{}
To begin with, we define an auxiliary variable as $\hat{\mathbf{w}}_{m+1} \triangleq \tilde{\mathbf{w}}_m -\eta \mathbf{g}_m$.
Hence, $\tilde{\mathbf{w}}_{m+1}$ is a unbiased estimation of $\hat{\mathbf{w}}_{m+1}$. Combining with \emph{Assumption~2} and the fact that $\nabla F(\mathbf{w}^*)\!=\!0$, we split the optimality gap $G_m\triangleq \mathbb{E}\left[ F(\tilde{\mathbf{w}}_{m+1})\right]-F(\mathbf{w}^*)$ into
\begin{small}
\begin{align}\label{ap1}
G_m\!\leq \frac{L}{2}\underbrace{ \mathbb{E} \left [ \left\Vert \tilde{\mathbf{w}}_{m\!+\!1}\!-\!\hat{\mathbf{w}}_{m\!+\!1} \right \Vert^2\right]}_{A_1}+\frac{L}{2}\underbrace{\mathbb{E} \left [ \left\Vert \hat{\mathbf{w}}_{m\!+\!1}\!-\!\mathbf{w}^* \right \Vert^2\right]}_{A_2}.
\end{align}
\end{small}
For the term $A_2$, according to \emph{lemma 2}  in \cite{importance}, we have
\begin{align}\label{a2}
A_2\!\leq\!\left( 1\!-\!\eta \mu \right)\mathbb{E}\left[\left \Vert \tilde{\mathbf{w}}_m \!-\!\mathbf{w}^*\right \Vert^2 \right]\!+\! \eta^2 \mathbb{E}\left [ \Vert\nabla F(\tilde{\mathbf{w}}_m)\Vert^2 \right ].
\end{align}
For the term $A_1$ in digital case, it is bounded by
\begin{small}
\begin{align}\label{a1}
A_1\!=&\eta^2 \mathbb{E}\left [\left \Vert \hat{\mathbf{g}}_{m,\text{D}} \!-\!\mathbf{g}_m \right \Vert^2\right]\leq \eta^2\underbrace{ \sum_{k=1}^K \alpha_k \mathbb{E}\left[\left \Vert \frac{\chi_k\xi_{k,\text{D}}}{ r_k} \mathcal{Q}(\mathbf{g}_m^k)\!-\!\mathbf{g}_{m}^k\right \Vert^2 \right ]}_{B_1}\nonumber \\
&+\eta^2\underbrace{ \sum_{k=1}^K \alpha_k \mathbb{E}\left[\left \Vert \mathbf{g}_{m}^k-\sum_{i=1}^K \alpha_i\mathbf{g}_m^i \right \Vert^2 \right ]}_{B_2}.
\end{align}
\end{small}
In analog case, the result only differs in term $B_1$. For the common part $B_2$, by expanding the square term, we have
\begin{align}\label{b2}
B_2=\sum_{k=1}^K \alpha_k  \mathbb{E}\left[\left \Vert \nabla F_k(\tilde{\mathbf{w}}_m)\right \Vert^2 \right]-\mathbb{E}\left[\left \Vert \nabla F(\tilde{\mathbf{w}}_m)\right \Vert^2 \right].
\end{align}
For digital case, according to \cite{quan2}, the variance of quantization error is bounded as
\begin{align}
\mathbb{E}\left [\left \Vert\mathcal{Q}(\mathbf{g}_m^k)-\mathbf{g}_m^k\right \Vert^2 \right ] &\leq \frac{d}{4}\left(\frac{ g_{\max}-g_{\min}}{2^b-1}\right)^2\leq \phi(b).
\end{align}
Then, $B_1$ is bounded by
\begin{align}\label{b1}
B_1\!\leq \! \sum_{k=1}^K \frac{\phi(b) \alpha_k}{p_k r_k}\!+\!\sum_{k=1}^K \! \alpha_k \! \left( \frac{1}{p_k r_k}\!-\!1\right)B_3,
\end{align}
where $B_3\triangleq \mathbb{E}\left[ \left \Vert \nabla F_k(\tilde{\mathbf{w}}_m) \right\Vert^2\right]$. For analog case, we have
\begin{align}
B_1\!=\!\sum_{k=1}^K\! \alpha_k \mathbb{E}\left[ \left(\frac{\chi_k\xi_{k,\text{A}}}{r_k}\!-\!1\right)^2\right]\! B_3
\! +\! \mathbb{E}\left[\Vert \bar{\mathbf{z}}_m \Vert^2 \right],
\end{align}
which can be calculated according to the statistics of $\xi_{k,\text{A}}$ in \cite{imperfect}.
Next, we further bound $B_3$ by
\begin{align}\label{final}
B_3\leq 2 L^2 \mathbb{E}\left[ \left \Vert \tilde{\mathbf{w}}_m-\mathbf{w}^* \right \Vert^2 \right ] +2L^2 \delta^2.
\end{align} 
Detailed procedure is omitted due to space limitations.
Combining all the results in (\ref{ap1})-(\ref{final}), we complete the proof.

\end{document}